\documentclass[journal,twoside,print,9pt]{ieeecolor}
\usepackage{generic}
\usepackage{verbatim}

\usepackage{etoolbox}
\makeatletter
\@ifundefined{color@begingroup}%
  {\let\color@begingroup\relax
   \let\color@endgroup\relax}{}%
\def\fix@ieeecolor@hbox#1{%
  \hbox{\color@begingroup#1\color@endgroup}}
\patchcmd\@makecaption{\hbox}{\fix@ieeecolor@hbox}{}{\FAILED}
\patchcmd\@makecaption{\hbox}{\fix@ieeecolor@hbox}{}{\FAILED}

\def\BibTeX{{\rm B\kern-.05em{\sc i\kern-.025em b}\kern-.08em
    T\kern-.1667em\lower.7ex\hbox{E}\kern-.125emX}}
\usepackage{graphicx,epstopdf,multirow,booktabs}
\usepackage{amsmath,amssymb,amsthm, bm}
\usepackage{cases}
\usepackage[T1]{fontenc}
\usepackage{hyperref}
\allowdisplaybreaks
\hypersetup{
     colorlinks =true,
     linkcolor = blue,
     filecolor = blue,
     citecolor = magenta,      
     urlcolor = magenta,
     }
\usepackage{setspace}
\usepackage{siunitx}
\usepackage{csvsimple}
\usepackage{url}
\usepackage[scaled]{helvet}
\usepackage[T1]{fontenc}

\usepackage[style=ieee,doi=false]{biblatex}
\AtBeginBibliography{\footnotesize}

\let\labelindent\relax
\usepackage{enumitem, xspace}

\usepackage{accents}

\newcommand{\NR}{Newton-Raphson\xspace}
\newcommand{\GF}{\mathcal{G}\mathcal{F}}

\theoremstyle{remark}
\newtheorem{remark}{Remark}
\newtheorem{definition}{Definition}
\newtheorem{assumption}{Assumption}
\newtheorem{proposition}{Proposition}
\newtheorem{lemma}{Lemma}
\newtheorem{theorem}{Theorem}

\def\d{\partial}

\def\mR{\mathbb{R}}
\def\domP{\mathcal{P}}

\begin{document}

\title{Numerical Solution of the Steady-State Network Flow Equations for a Non-Ideal Gas}

\author{Shriram Srinivasan$^{\dagger}$, Kaarthik Sundar$^{*}$, Vitaliy Gyrya$^{\dagger}$, Anatoly Zlotnik$^{\dagger}$
\thanks{$^{\dagger}$Applied Mathematics and Plasma Physics Group, Los Alamos National
Laboratory, Los Alamos, New Mexico, USA. E-mail: \texttt{\{shrirams,vitaliy\_gyrya,azlotnik\}@lanl.gov}}\;
\thanks{$^{*}$Information Systems and Modeling Group, Los Alamos National
Laboratory, Los Alamos, New Mexico, USA. E-mail: \texttt{{kaarthik}@lanl.gov}}\;
\thanks{The authors acknowledge the funding provided by LANL’s Directed Research and Development (LDRD) project: ``20220006ER: Fast, Linear Programming-Based Algorithms with Solution Quality Guarantees for Nonlinear Optimal Control Problems'' and by the U.S. Department of Energy's Advanced Grid Modeling (AGM) projects Joint Power System and Natural Gas Pipeline Optimal Expansion and Dynamical Modeling, Estimation, and Optimal Control of Electrical Grid-Natural Gas Transmission Systems. The research work conducted at Los Alamos National Laboratory is done under the auspices of the National Nuclear Security Administration of the U.S. Department of Energy under Contract No. 89233218CNA000001.}\; 
}

\maketitle

\begin{abstract}
We formulate a steady-state network flow problem for non-ideal gas that
relates injection rates and nodal pressures in the network  to flows in pipes.
For this problem, we present and prove a theorem on uniqueness of generalized solution
for a broad class of non-ideal pressure-density relations that satisfy a monotonicity property. 
Further, we develop a \NR  algorithm for numerical solution of the steady-state problem, which is made possible by a systematic non-dimensionalization of the equations.  The developed algorithm has been extensively  tested on benchmark instances and shown to converge robustly to a generalized solution.
Previous results  \cite{de2000gas, ojha2017solving,Singh2019,Singh2020}, indicate that the steady-state network flow equations for an ideal gas are 
difficult to solve by the \NR method because of its extreme sensitivity to the initial guess. 
In contrast, we find that non-dimensionalization of the steady-state problem is key to robust convergence of  the \NR  method.
We identify criteria based on the uniqueness of solutions under which the existence of a non-physical generalized solution found by a non-linear solver implies non-existence of a physical solution, i.e., infeasibility of the problem. 
Finally, we compare pressure and flow solutions based on ideal and non-ideal equations of state to demonstrate the need to apply the latter in practice. 
The solver developed in this article is open-source and is made available for both the academic and research communities as well as the industry.
\end{abstract}

\begin{IEEEkeywords}
steady-state, network flow equations, non-ideal gas, \NR, compressibility factor
\end{IEEEkeywords}
\IEEEpeerreviewmaketitle

\def\d{\partial}

\section{Introduction} \label{sec:intro}

Over the past decade, the increase in new gas-fired electricity generation and concurrent growth of renewable energy sources in the electric power grid has led to the increasing reliance on natural gas to fuel base generation and to balance out load fluctuations \cite{li2008interdependency,roald2020uncertainty}. 
Moreover, natural gas also serves as a critical energy source for domestic heating in the U.S. \cite{rios2015optimization}. 
In the United States, natural gas consumption reached a record high of 85 billion cubic feet per day in 2019 \cite{EIA-1}. 
Furthermore, data from the U.S. Energy Information Administration indicates that natural gas accounts for the largest share of generation fuel  since surpassing coal on an annual basis in 2016 \cite{EIA-2}. 
Natural gas production sites are usually situated in remote locations and are geographically well separated from consumption locations. 
An infrastructure network of pipelines is used to transport the natural gas from gathering and processing facilities to consumers. 
Natural gas system operators work to ensure safe and reliable transport of gas through pipelines, and utilize various decision support tools for system design and operation that include simulation and  optimization. 

Many optimization problems for gas pipeline systems have been formulated and examined.  
Such problems aim to minimize or maximize an objective function involving a cost or performance index subject to constraints that represent the physics of transient or steady-state natural gas flows.  
Examples include optimal network expansion \cite{Zheng2010, borraz2016convex,sundar2021robust}, integrated gas-electric system operations \cite{Raheli2021,He2018,roald2020uncertainty}, day-to-day operations of natural gas pipeline with \cite{hari2021operation} and without storage \cite{zlotnik2019optimal,wu2017adaptive}, state and parameter estimation \cite{Behrooz2015,Jalving2018,sundar2018state,sundar2019dynamic}, and compressor power minimization \cite{zlotnik2015optimal,Deng2019}.  

The usual operating dynamics of natural gas flow through pipelines in a physical regime without waves and shocks can be adequately described by a system of coupled partial differential equations (PDE) in density, pressure and mass flux variables \cite{GyryaZlotnik2017}.  
These equations represent balance of mass and momentum, where momentum dissipation is modeled using the Darcy-Wiesbach friction approximation, which is quadratic in velocity. It is well-understood that an isothermal approximation is adequate in this regime \cite{osiadacz01}.  The equation of state is specified as either a linear (ideal)  or nonlinear (non-ideal) model that relates gas pressure and density. 
Under steady-state conditions and ideal gas assumptions, the system of coupled PDEs reduce to a non-linear system of algebraic equations that relate pressure and mass flow values throughout the pipeline network.  
While gas flows in real pipelines do in general undergo temporal variation, mid-term and long-term planning questions are in practice often examined using steady-state models, even though simulation tools for transient models are available \cite{koch15}. On the one hand, this is because the network sizes currently solvable by transient methods are small.  On the other hand, in mid- or long-term planning, future nomination profiles and their time-dependence may not be known. Instead, fictitious future nominations are considered, where load flows and external conditions are assumed to be constant over a fixed time period. Moreover, when restricted to steady-state flows, nominations can be aggregated over a day, for example. 
Furthermore, many optimization studies for gas pipeline networks that feature slow variations over a time horizon  simplify flows to sequences of steady-state problems at discrete time instants within the time horizon \cite{gugat2021transient}.

In this study, we derive the corresponding steady-state equations for a particular case of a non-ideal pressure-density relation.  
These equations capture the intuitive fact that as natural gas flows along a pipeline, the pressure drops non-linearly along its length.  
As a result of that property, pressures in a pipeline are maintained within a certain range of required values using compressors.  
The effects of compressor stations can be modelled as either multiplicative  \cite{zlotnik2015optimal} or additive \cite{sundar2021robust,vuffray15cdc} factors, which scale up the pressure of the gas between the station inlet and outlet. 

The gas flow ($\GF$) problem examined in this study is as follows. 
\underline{Given} (i) and (ii):
\noindent  
(i) Either the pressure 
or the mass flow rate of injection/extraction at every node; and 
(ii) compression ratios for each compressor in the network;
\noindent  
\underline{find} the injection/extraction mass flow rates (at the nodes where pressure is specified),
the pressures (at the nodes where injection/extraction mass flow rates are specified) in the system, and the mass flows across all the pipes and compressors that satisfy the non-linear steady-state equation system governing the flow of natural gas in the entire network.
Methods for rapidly and accurately solving this problem at scale \cite{kenworgu08phd,koch15} are critical for evaluating feasible flow configurations and capacities of natural gas transmission pipeline networks .

The $\GF$ problem has been addressed previously under the ideal gas assumption, although in the high pressure conditions of transport pipeline flow, natural gas will not behave according to this assumption \cite{GyryaZlotnik2017}. 
Recent studies have proposed Mixed-Integer Quadratically Constrained Quadratic Program (MI-QCQP) \cite{Singh2020}, Mixed-Integer Second Order Cone Program (MI-SOCP) \cite{Singh2019}, and Semi-Definite Program (SDP) approaches \cite{ojha2017solving} to solve the $\GF$ problem. 
All these methods are computationally expensive and do not scale well for larger networks.
In particular, \cite{Singh2020} develops a convex relaxation for the ideal gas case only, but  is proved to be exact under some assumptions on the networks, namely, a single slack node,  no cycles with shared edges, and no circulatory flows in cycles.
The \NR formulation has no such restrictions, and we show that it works successfully with multiple slack nodes and overlapping cycles.
Moreover, an MI-QCQP is an NP-hard problem, and cannot scale well for large instances. This is also reflected in the reported computation times which are $10^2-10^3$ times that of the \NR performance for the GasLib-40 instances in this study.
Another study \cite{de2000gas} has proposed a primal-dual based method to solve $\GF$ problem on networks without compressors, which is huge limitation of the method.
These studies all conclude that a standard \NR solver does not converge without specific initialization, and performs poorly in comparison to the respective methods proposed therein.  One study \cite{Singh2020} states that an approach that makes \NR work effectively for the $\GF$ problem, and independently of its initialization, remains an open problem of interest.  We aim to address this issue comprehensively in this study.

This article makes four contributions.  
First, we derive the non-linear steady-state equations for the $\GF$ problem assuming a general equation of state for the gas, and consider a particular form that is a generalization of the ideal gas equation. 
Second, we demonstrate that uniqueness results derived for the system of equations under the ideal gas assumption \cite{Singh2019}, \cite{Singh2020} generalize under weaker hypotheses to any non-ideal gas based solely on the monotonicity property of the pressure drop along the pipes \cite{misra2020monotonicity}. 
Third, we demonstrate that the standard \NR algorithm with random initialization reliably converges to the unique solution of the $\GF$ problem when the steady-state equations have been suitably non-dimensionalized.  
We claim that this result resolves the open research question (identified in \cite{Singh2019}) of developing an \NR type  $\GF$ solver that performs well for large pipeline networks without requiring special procedures for initialization. 
Finally, we present extensive numerical simulations on benchmark test instances that 
(i) compare the $\GF$ solutions obtained for ideal and non-ideal equations of state; 
(ii) show the effectiveness of the proposed solver for the $\GF$ problem in terms of computation time and the number of iterations required for convergence; and 
(iii) show the importance of appropriate non-dimensionalization to enable convergence of the algorithm. 

The rest of the article is organized as follows: Section \ref{sec:ss-pipe} develops the equations that govern the steady flow of a non-ideal gas for a single pipe. Section \ref{sec:non-dimensionalization} presents a generic technique to non-dimensionalize the steady-state governing equations developed in Section \ref{sec:ss-pipe}, and is followed by the formulation of the $\GF$ problem on a pipeline network in Section \ref{sec:gf-network}. Section \ref{sec:gf-network} also includes a discussion of uniqueness of the $\GF$ problem solution for a non-ideal gas. Section \ref{sec:algo} presents a discussion of the \NR algorithm we use to solve the $\GF$ problem. Extensive computational experiments are presented in Section \ref{sec:results} that corroborate the effectiveness of the \NR algorithm on a wide class of standard benchmark instances and finally, the article concludes with Section \ref{sec:conclusion}. 

\section{Steady gas flow in a pipe} \label{sec:ss-pipe}
The adiabatic flow of compressible gas in a single pipeline is described by the Euler equations in one dimension \cite{thorley87}. 
Inertial terms for long pipelines can be ignored, which leads to the following equations (see \cite{GyryaZlotnik2017}) that represent conservation of mass and momentum balance:
\begin{subequations}
\begin{flalign}
  \frac{\d\rho}{\d t}+\frac{\d \varphi}{\d x} & = 0, 
  \label{eq:euler1a} \\
  \frac{\d\varphi}{\d t} + \frac{\d p}{\d x}  & = 
  -\frac{\lambda}{2D}\frac{\varphi |\varphi|}{\rho}, 		
  \label{eq:euler1b}
\end{flalign}
\label{eq:euler1}
\end{subequations}
where $\rho$ is the density, $p$ is pressure, $\varphi=\rho v$ is mass flux, and $v$ is velocity of the gas. 
The additional parameters are the friction factor $\lambda$ and diameter $D$ of the pipe.  
The term on the right hand side of \eqref{eq:euler1b} quantifies the energy dissipation caused by friction in turbulent flow. We have ignored effects related to change in elevation for simplicity of presentation, and also because natural gas pipelines ( e.g., of the continental United States) are nearly horizontal on the large scale. 
The mass and momentum conservation equations are supplemented with the equation of state (EoS), which relates the density $\rho$ and the pressure $p$ of the gas as
\begin{align}\label{eq:eos}
    p = Z(p, T) R_g T \rho
    \quad
    \text{or}
    \quad
    \rho(p)
    =
    \frac{p}{Z(p, T) R_g T},
\end{align}
where $R_g$ is the specific gas constant, $T$ is gas temperature, and $Z$ is the compressibility factor that may depend on pressure and temperature for a non-ideal gas.
The compressibility factor $Z(p,T)$ is typically given as a formula with parameters that have been fitted to measured data obtained during early engineering studies \cite{benedict1940empirical,elsharkawy2004efficient}. 
A widely-used formula for $Z(p, T)$ is the CNGA EoS \cite{menon2005gas}, given by
\begin{flalign}\label{eq:cnga}
    Z(p, T) = \frac 1{b_1 + b_2 p} 
\end{flalign}
where, $b_1$ and $b_2$ are gas and temperature-dependent constants. The values of $b_1$ and $b_2$ are given by the following expressions:
\begin{align}
  b_1 = 1 + \left( \dfrac{p_{atm}}{6894.75729}\right) 
  \left( \dfrac{a_1 10 ^ {a_2 G}}{(1.8T) ^ {a_3}} \right) ~~(\text{\si{unitless}}),\label{eq:b1}\\
    b_2 = \left(\dfrac{1}{6894.75729}\right)
    \left(\dfrac{a_1 10 ^ {a_2 G}}{(1.8T)^{a_3} } \right) ~~(\text{\si{\per\pascal}}). \label{eq:b2}
\end{align}
Here, $b_1$ and $b_2$ are calculated in terms of other non-dimensional constants $a_1 = 344400$, $a_2 = 1.785$, $a_3 = 3.825$, specific gravity of natural gas $G = 288.706$ and atmospheric pressure $p_{atm} = 101350\; \si{\pascal}$. 

Here we are interested in  steady-state solutions of Eq. \eqref{eq:euler1} under isothermal conditions.
Note that since the temperature is constant, we shall write $p(\rho)$ and $\tfrac{dp}{d\rho}$ for simplicity instead of $p(\rho, T)$ and $\tfrac{\partial p}{\partial \rho}$.
Thus, setting time derivative terms to zero, the conservation of mass equation yields
\begin{flalign}
    \dfrac{d \varphi}{d x} = 0 
    \quad \implies \quad 
    \varphi = \text{constant}. \label{eq:constant-flux}
\end{flalign}
The above equation indicates that in a steady-state regime, the mass flux across the cross-sectional area of the pipe is a constant throughout the length of the pipe.
Assuming a constant mass flow $f = A\varphi$ in a pipe with cross-sectional area $A$, the momentum balance equation can be rewritten in terms of mass flows instead of mass fluxes as
\begin{align}
    \dfrac{dp}{dx} 
    = 
    -\dfrac{\lambda}{2D} \dfrac{\varphi |\varphi|}{\rho(p)} 
    = 
    - \dfrac{\lambda}{2DA^2} \dfrac{f |f|}{\rho(p)}. \label{eq:ode}
\end{align}

In order to account for the effect of a nonlinear EoS on pipe flow, we define a potential function $\Pi(p)$
as a solution of the differential equation
\begin{flalign} \label{eq:pi-def}
    \frac{d \Pi(p)}{d p} :=  \rho(p)
    \ \ \implies \ \ 
    \Pi(p) = \Pi(p_0) + \int_{p_0}^p \rho(\tilde p)\ d\tilde p,
\end{flalign}
which is defined uniquely up to an additive constant.

Multiplying Eq. \eqref{eq:ode} by $\rho(p)$ and substituting \eqref{eq:pi-def} results in
\begin{align}\label{eq:meaning of potential}
    \frac{d \Pi(p)}{d p} \dfrac{dp(x)}{dx} 
    =
    \dfrac{d \Pi(p(x))}{dx} 
    =
    - \dfrac{\lambda f |f|}{2DA^2}. 
\end{align}
Integrating \eqref{eq:meaning of potential} along the length of the pipe with the end-point data labelled with subscripts  
$1$ and $2$ and constant mass flow $f$ directed from point $x_1$ to $x_2$, the end-point pressures $p_1$ and $p_2$ are related by  
\begin{flalign}
    \Pi(p_2) - \Pi(p_1)  = - L \dfrac{\lambda f |f|}{2DA^2}, \label{eq:gen-non-ideal-physics-flows}
\end{flalign}
where $L=|x_2-x_1|$ is the length of the pipe. 
If $f$ is negative in Eq. \eqref{eq:gen-non-ideal-physics-flows}, 
then the direction of flow is from $x_2$ to $x_1$.

The first derivative of the potential function $\Pi$ is the density, as apparent from the definition Eq. \eqref{eq:pi-def}. The formula for the second derivative  of $\Pi$ turns out to  be the product of density (positive) and a physical quantity that is defined as isothermal compressibility of  the gas (which is positive) and hence $\Pi"$ is positive for positive densities according to
%
%
\begin{flalign}\label{eq:sound-speed_b}
\Pi''(p) = \frac{d \rho(p)}{d p} = \frac{\partial \rho(p, T)}{\partial p} > 0.
\end{flalign}

\subsection{Simplifications for specific equations of state} \label{subsec:specific_eos}

We now derive the form of Eq. \eqref{eq:gen-non-ideal-physics-flows} for two particular equations of state that are of interest in the context of natural gas.

For the CNGA EoS, the dependence between $\rho$ and $p$ is precisely defined  by combining Eq. \eqref{eq:eos} and \eqref{eq:cnga} as 
\begin{flalign}
    \rho = \frac{b_1 p + b_2 p^2}{R_gT}. \label{eq:rho(p)}
\end{flalign}

We denote by $a > 0$ the fixed quantity $a = \sqrt{R_gT}$. Note that for an ideal gas, $\frac{d \rho(p)}{d p} = \frac 1{R_gT}$. Rewriting Eq. \eqref{eq:rho(p)} for isothermal conditions using $a$, we obtain
\begin{flalign}
    \rho = \frac{b_1 p + b_2 p^2}{a^2}. \label{eq:rho-a}
\end{flalign}
For the above form of the equation of state $\rho(p)$,  
Eq. \eqref{eq:gen-non-ideal-physics-flows} simplifies to
\begin{flalign}
\frac{b_1}{2}(p_2^2 - p_1^2) + \frac{b_2}{3}( p_2^3 - p_1^3) = - \dfrac{L\lambda a^2}{2DA^2} f |f|. \label{eq:non-ideal-physics-flows}
\end{flalign}

For an ideal gas, the EoS in Eq. \eqref{eq:eos} is linear with compressibility factor $Z = 1$. In that case, the equation that governs the flow of gas through the pipe under steady-state conditions can be obtained by setting $b_1 = 1$ and $b_2 = 0$ in Eq. \eqref{eq:non-ideal-physics-flows}, i.e., 
\begin{flalign}
p_2^2 - p_1^2 = - \dfrac{L\lambda a^2}{DA^2} f |f|. \label{eq:ideal-physics-flows}
\end{flalign}
The existing gas flow solvers in the literature use the above equation to develop computational methods to solve the $\GF$ problem. The form of Eq. \eqref{eq:ideal-physics-flows}  was exploited by the authors in \cite{Singh2019,Singh2020,borraz2016convex}  to develop a Mixed-Integer Second-Order Cone relaxation. In the next section, we present a systematic technique to non-dimensionalize Eq. \eqref{eq:gen-non-ideal-physics-flows}, which will later play a key role in developing solution techniques to solve the $\GF$ problem on a network of pipelines. 

\section{Non-dimensionalization of the governing equations for a single pipe} \label{sec:non-dimensionalization}

Non-dimensionalization of the governing equations is essential to avoid an ill-scaled problem when looking for  a numerical solution of the gas flow equations. 
This is crucial because operating pressures of natural gas pipelines are in the range \SIrange{3}{7}{\mega\pascal}  and the injection and withdrawal rates of natural gas at the production and consumer locations can be in the range  \SIrange{5}{300}{\kilogram\per\second}. Thus, the flow and pressure variables differ by orders of magnitude.  
Furthermore, it has been observed that using the standard \NR algorithm for the governing equations does not lead to convergence \cite{Singh2019,Singh2020} even when the ideal gas EoS is applied.  In the subsequent paragraphs, we present a generic technique to non-dimensionalize the governing equation in Eq. \eqref{eq:gen-non-ideal-physics-flows}. Nondimensional quantities  denoted by an overbar are defined by scaling a physical quantity with an appropriate nominal value.

We first select the nominal length, pressure, density, and velocity as $l_0,  p_0, \rho_0$, and $v_0$ respectively, and then set the nominal mass flux to be $\varphi_0 = v_0 \rho_0$. 
To obtain the non-dimensional equation in terms of mass flows and a non-dimensional function $\bar{\Pi}$ defined as in Eq. \eqref{eq:pi-def}, we set the nominal area to $A_0=1$, so that $\bar A = A$ and the non-dimensional mass flow is $\bar f = \bar A \bar{\varphi}$. 
We chose the base value $A_0 =1$ so that the nominal mass flux and nominal mass flows have identical values,  though it is possible to choose $A_0$ to be some other value.
Then setting $\bar x = x/l_0$, $\bar p = p/p_0$, $\bar{\varphi} = \varphi/\varphi_0$, and $\bar D = D/l_0$, leads $\bar{f}$ to be a constant, and Eq. \eqref{eq:ode} reduces to 
\begin{flalign}
\bar \rho \dfrac{d\bar p}{d\bar x} = -\frac{\lambda}{2\bar D \bar A^2} \bar{f} | \bar{f} |  \left(\frac{\rho_0 v_0^2}{p_0}\right). \label{eq:ode-nd-0}
\end{flalign}
In the above equation, if we let $\mathcal M = v_0/a$ denote the Mach number of the nominal flow velocity and $\mathcal C = p_0/\left(\rho_0 a^2\right)$ be a constant analogous to the Euler number, then Eq. \eqref{eq:ode-nd-0} can be rewritten as 
\begin{flalign}
\bar \rho \dfrac{d\bar p}{d\bar x} = -\frac{\mathcal M^2}{\mathcal C} \frac{\lambda}{2\bar D \bar A^2} \bar{f} | \bar{f} |. \label{eq:ode-nd}
\end{flalign}
Then if we define
\begin{flalign}
\bar\Pi(\bar p) \triangleq \bar\Pi(\bar p_0) + \int_{\bar p_0}^{\bar p} \bar \rho(\bar p) d\bar p, \label{eq:pi-def-nd}
\end{flalign}
we obtain the nondimensional counterpart to Eq. \eqref{eq:gen-non-ideal-physics-flows} as
\begin{flalign}
\bar\Pi(\bar p_2) - \bar\Pi(\bar p_1) = -\frac{\mathcal M^2}{\mathcal C} \frac{\lambda \bar L}{2\bar D \bar A^2} \bar{f} | \bar{f} |, \label{eq:gen-non-ideal-physics-flows-nd}
\end{flalign}
where  $\bar L = L/l_0$ is the non-dimensional length of the pipe.

We compare CNGA and ideal gas EoS as representatives of two different forms of the EoS in our numerical studies.
After converting the variables in the density-pressure relationship in Eq. \eqref{eq:rho-a} to dimensionless quantities, we obtain
\begin{flalign}
\bar{\rho} = \bar b_1\bar p + \bar b_2 \bar p^2,  \label{eq:rho-nd}
\end{flalign}
where $\bar b_1 = \mathcal C b_1$ and $\bar b_2 = \mathcal C p_0 b_2$.  Substituting the result into Eq. \eqref{eq:gen-non-ideal-physics-flows-nd} yields
\begin{flalign}
\frac{\bar b_1}{2}(\bar p_2^2 - \bar p_1^2) + \frac{\bar b_2}{3}(\bar p_2^3 - \bar p_1^3) = - \frac{\mathcal M^2}{\mathcal C} \dfrac{\lambda \bar L}{2\bar D \bar A^2} \bar f |\bar f|. \label{eq:non-ideal-physics-nd}
\end{flalign}

As expected, setting $\bar{b}_1 = 1$ and $\bar{b}_2 = 0$ in Eq. \eqref{eq:non-ideal-physics-nd} yields the non-dimensionalized equations for an ideal gas. As an aside, note that the nominal quantities are usually chosen to ensure that the non-dimensional variables have similar orders of magnitude.  If the intent however is to have the transformed equations appear without the Euler and Mach numbers \cite{sundar2018state,hari2021operation}, one can choose $v_0 =a$ and $\rho_0 = p_0/a^2$ to ensure that $\mathcal{M} = \mathcal{C} = 1$.

Note that while we are free to choose the values of $l_0, p_0, \rho_0$, and $v_0$, we aim for a choice that will lead to good scaling in the $\GF$ problem and superior performance of any solver used for it on large pipeline networks. In the next paragraph, we present guidelines for choosing $l_0$, $p_0$, $\rho_0$ and $v_0$ which have been used for re-scaling the problem and  subsequently lead to successful convergence of the $\GF$ solver in the computational experiments in our study. \emph{In subsequent discussions, for ease of presentation, we shall drop the overbar that designates non-dimensional quantities with the understanding that all quantities are dimensionless.}

\subsection{Guidelines for choosing nominal values} 

We first start by choosing a value of nominal length $l_0$ that is representative of the typical length of a single pipe, which is usually in the range of \SIrange{1000}{10000}{\meter}. 
Any value of nominal length that occurs in the range of pipeline lengths for a given network works well.
The value of nominal pressure is usually decided based on the range of operating pressure in the pipeline network.  For transmission networks, this range is \SIrange{3}{7}{\mega\pascal}. 
For a given network, a non-zero slack pressure can be used  as the nominal value.
Since  $a = \sqrt{R_gT}$  has the dimension of velocity, the nominal velocity is conveniently expressed as a factor of it. The nominal velocity $v_0$ is chosen such that the quantity $\rho_0 v_0$ is of the same order of magnitude as some nodal injection or withdrawal.
Once $l_0, p_0, v_0$ are selected, the nominal density, mass flux, and mass flow rates are computed as
\begin{flalign}
\rho_0 = \frac{p_0}{a^2},
\varphi_0 = \rho_0 v_0, \text{ and } f_0 = \varphi_0 A_0, \label{eq:nominal-quantities}
\end{flalign}
where $A_0 = 1$. 
The choice of nominal scaling factors is subjective, and there may be other ways to identify scaling factors that ensure the $\GF$ problem scales well in the sense that non-dimensional variables have similar orders of magnitude. 
However, the procedure we have described here to choose sensible nominal values is general and uses nothing more than the data given for a network.
In the next section, we present the formulation for the $\GF$ problem for a pipeline network with multiple pipes and compressors. 

\section{Gas flow problem on a network} \label{sec:gf-network}
A gas pipeline network consists of pipes, junctions, and compressors. A junction is a physical location where one or more pipes and compressors come together. In previous studies, we find two equivalent ways of defining a compressor. In the \emph{endpoint} representation  \cite{GyryaZlotnik2017,Singh2020}, compressors are associated with the ends of pipes and each compressor can provide a prescribed pressure boost ratio. The notion of a trivial compressor (a compressor with a pressure boost ratio of 1) is used, and this is equivalent to having no compressor at the end of the pipe. Thus, without loss of generality, one can assume that each pipe has two controllers, with one at each end.  That mechanism enables modeling of any type of pressure or flow boundary condition at either end of a pipe, which is useful for representing valve actions in addition to compression in transient simulation. Alternatively, in the \emph{node-connecting} representation \cite{hari2021operation,borraz2016convex}, compressors are associated with a pair of junctions and could be viewed as a pipe with zero length and prescribed pressure boost ratio. This representation simplifies the overall equations governing the flow of natural gas through networks of pipes and compressors in the context of steady-state modeling.

Hence, in this article,  we use the node-connecting representation for algorithmic simplicity and ease of exposition. Note that it is possible to perform a conversion between these two representations of a compressor. 
The gas consumption caused by compressors powered using gas extracted from the network is actually an extremely small fraction (less than 0.005 \cite{Herty2007}) of the incoming flow, and hence neglected in the modelling. Moreover, compressor stations are typically located several dozen kilometers apart, and thus they are not numerous enough to make the cumulative consumption significant either.

We begin by setting up our notation. We let $G = (N, P \cup C)$ denote the pipeline network, where $N$, $P$ and $C$ are the set of junctions, pipes and compressors respectively. A pipe $(i, j) \in P$ connects the junctions $i$ and $j$. Also, a compressor $(i, j) \in C$ connects junctions $i$ and $j$, which we assume are  geographically co-located. For each junction $i \in N$ we let $p_i$ and $q_i$ denote the non-dimensional pressure and injection at $i$ respectively. If $q_i < 0$, this indicates that gas is being withdrawn from the network at $i$ at a rate of $|q_i|$. For each pipe $(i, j)$ that connects junctions $i$ and $j$, we let $f_{ij}$ denote the non-dimensional steady flow through the pipe and $\lambda_{ij}$, $A_{ij}$, $D_{ij}$ and $L_{ij}$ denote the friction factor, cross-sectional area, diameter and length of the pipe respectively. it is assumed that all  parameters of the pipe are non-dimensionalized. Finally, for each compressor $(i, j) \in \mathcal C$, we use $\alpha_{ij}$ to denote the pressure boost ratio or compressor ratio, and use $f_{ij}$ to denote the flow through the compressor.  It is assumed that the direction of flow is in the direction of pressure boost, i.e., $i \rightarrow j$, and hence $f_{ij} > 0$ for a compressor $(i, j)$. 

Finally, the set of junctions $N$ in the network is partitioned into two mutually disjoint sets consisting of the slack junctions $N_s$ and non-slack junctions $N_{ns}$. The slack junctions have pressures specified, and non-slack junctions have gas injection into the network specified. Before we formulate the $\GF$ problem for the gas network $G$, we enumerate the assumptions we make about the network.

\noindent \textit{Assumptions about the network and the provided data:} 
\begin{enumerate}[label=(A\arabic*)]
    \item There is at least one slack node, i.e., $|N_s| \geq 1$.\label{assumption:slacks}
    \item For all the compressors in $C$, the pressure boost or compression ratio is known a-priori. 
    \item If there are multiple slack nodes ($|N_s| > 1$), a path connecting two slack nodes must consist of at least one pipe. \label{assumption:path-pipe}
    \item Any cycle must consist of at least one pipe. \label{assumption:cycle-pipe}
\end{enumerate}
With these assumptions, we now proceed to record the system of equations that need to be solved. 
The rationale behind \ref{assumption:path-pipe} and \ref{assumption:cycle-pipe}  will be made clear later, but it is linked to the fact that except for pipes, other edge elements have their flows and vertex pressures decoupled.

\subsection{Governing equations for the network} \label{subsec:network-gf}

The steady mass flow rate of gas $f_{ij}$ in each pipe $(i, j) \in P$ and the pressures at the ends $i$ and $j$ have to satisfy Eq. \eqref{eq:gen-non-ideal-physics-flows-nd}. 
This equation is rewritten using the notation presented in the previous paragraphs as follows:
\begin{flalign}
    \Pi\left(p_i\right) - \Pi\left(p_j\right) = \beta_{ij} f_{ij} |f_{ij}| 
    \label{eq:pipe}
\end{flalign}
where $\beta_{ij} \triangleq (\mathcal M^2/\mathcal C) \cdot \lambda_{ij} L_{ij} / (2 D_{ij} A_{ij}^2)$ is the effective resistance of the pipe. Here, the convention is that when gas flows from junction $i \rightarrow j$ (resp. $j \rightarrow i$), then $f_{ij}$ is positive (resp. negative). 

Each compressor $(i, j) \in C$ is associated with the pressure boost equation 
\begin{flalign}
p_j = \alpha_{ij} p_i,
\qquad \alpha_{ij}\geqslant 1. \label{eq:compressor}
\end{flalign}
As mentioned in the previous section, the mass flow through the compressor $f_{ij}$ is along the direction of compression, hence $f_{ij}$ is strictly non-negative.  In this model, the mass flow through the compressor is governed by the nodal balance equations alone.

Finally, at each non-slack junction $i \in N_{ns}$ in the network, i.e., the junctions where the net injection into the system is specified, we have the flow balance equation. 
To formulate the flow-balance equations, we let $\bm A$ be a reduced edge incidence matrix of the graph $G$ 
of size $|N_{ns}| \times (|P|+|C|)$. The full edge incidence matrix 
$\bm A^{\mathrm{full}}$ of the graph $G$ is of size $|N| \times (|P|+|C|)$.

Each element of $\bm A$ and $\bm A^{\mathrm{full}}$  is defined as follows: 
\begin{flalign}
A_{ij} = \begin{cases}
-1 & \text{ if $e_j = (v_i, \cdot) \in P\cup C$, } \\
+1 & \text{ if $e_j = (\cdot, v_i) \in P \cup C$, } \\
0 & \text{ otherwise. }
\end{cases} \label{eq:incidence}
\end{flalign}

Let $\bm q \in \mathbb R^{|N_{ns}|}$ denote the vector of specified injections at the non-slack junctions,  $\bm q^{\mathrm{full}} \in \mathbb R^{|N|}$  be the full vector of injections that includes the unknown injections at the slack nodes while $\bm f \in \mathbb R^{(|P|+|C|)}$  be the vector of mass flows in the pipes and compressors. Given this notation, the nodal balance equations for the non-slack nodes are
\begin{flalign}\label{eq:nodal-balance}
    \bm A  \bm f = \bm q,
\end{flalign}
and the nodal balance for \emph{all} nodes reads
\begin{flalign}\label{eq:nodal-balance-all}
    \bm A^{\mathrm{full}}  \bm f = \bm q^{\mathrm{full}}.
\end{flalign}

In summary, the steady-state gas flow problem on the pipeline network takes the form

\begin{subnumcases} {\label{eq:GF}
\mathcal G \mathcal F:}  
\Pi(p_i) - \Pi(p_j) = \beta_{ij} f_{ij} | f_{ij}| \quad \forall (i, j) \in P , &  \label{eq:GFpipe}\\ 
p_j = \alpha_{ij} p_i \quad \forall (i, j) \in C, & \label{eq:GFcompressor} \\ 
\bm A \bm f = \bm q, & \label{eq:GFbalance}\\
p_i \text{ specified } \forall i \in N_s. & \label{eq:GFslack}
\end{subnumcases}

Note that any $ \bm f \in \mathbb R^{(|P|+|C|)}, \bm p \in \mathbb{R}^{|N_{ns}|}$  that satisfies Eq. \eqref{eq:GF} is mathematically a solution, even if it violates certain physical assumptions.  Moreover, any such solution to the steady-state flow problem satisfies the injection/extraction balance conditions
\begin{align}\label{eq:injection-extraction balance}
    \underset{i\in N_s}{\sum} q_i + \underset{j\in N_{ns}}{\sum} q_j = 0.
\end{align}
The condition \eqref{eq:injection-extraction balance} is a consequence of the flow balance conditions \eqref{eq:nodal-balance-all}
applied to all nodes in the $\GF$ problem.
Indeed, summing the conditions \eqref{eq:nodal-balance-all}
over all nodes yields  \eqref{eq:injection-extraction balance} because mass flow through each pipe appears exactly twice -- once with a plus sign and once with a minus sign, adding up to zero.
\begin{remark}
 In general, a given set of potentials will determine a unique set of mass flows, but these mass flows may or may not satisfy flow balance. Similarly, a given set of mass flows that satisfy flow balance along with a given slack pressure will determine a set of potentials only if the  slack pressure and mass flows are compatible. However, if a solution $(\bm{p}, \bm{f})$ exists for the system, then it is certainly true that  $\bm{p}$ can be determined from $\bm{f}$ and \emph{vice-versa} through \eqref{eq:GFpipe}.
\end{remark}
\begin{remark}
The previous remark suggests that uniqueness follows if Eq.~\eqref{eq:GFbalance} could be inverted, but this happens only for a tree network with $|N_s| = 1$, where the injectivity of $\bm A^{\mathrm{full}}$ for a tree in conjunction with the fact that $|N_s| = 1$ implies that $\bm A$ is invertible.
\end{remark}
The given data for a $\GF$ problem is a tuple $(\bm q, \bm p)$ where 
$\bm q \in \mathbb R^{|N_{ns}|}$ denotes the vector of specified injections and 
$\bm p \in \mathbb{R}^{|N_s|}$ the vector of slack pressures. 
The condition \eqref{eq:injection-extraction balance} constrains $\bm q^{\mathrm{full}}$, the full vector of injections at all the junctions. 
If $|N_s| =1$, the input conditions $\bm q$ together with the flow balance condition \eqref{eq:injection-extraction balance} determines all of $\bm q^{\mathrm{full}}$. 
However, if $|N_s| > 1$, it is not a priori clear if $(\bm q, \bm p)$ determines $\bm q^{\mathrm{full}}$ uniquely. This question will be addressed in Section~\ref{subsec:uniqueness}.
\subsection{Uniqueness of solution to the \texorpdfstring{$\GF$}{GF} problem on a network} \label{subsec:uniqueness}
The system of equations for the $\GF$ problem on the network has $|P| + |C|$ flow variables, $|N_{ns}|$ pressure variables, and the $|P| + |C| + |N_{ns}|$ total equations because the compressor ratios, slack junction pressures, and the non-slack junction injection values are known a-priori. For convenience, $|N_s|$ trivial slack pressure equations may be added to the system to have $|P| + |C| + |N|$ equations and variables.

Let $\domP\subseteq\mR$ be the domain of allowed pressure values $p$ for the $\GF$ problem.
If the potential $\Pi$ be strictly  increasing on this domain, i.e., for any $p$, $\hat p\in \domP$, $\Pi(\hat p)>\Pi(p)$ whenever $\hat p >p$, the monotonicity property for the potential $\Pi$ on this domain will allow us to prove uniqueness of a solution if it exists. Depending on our objective, we will be interested in various choices of the domain $\domP$.

A solution that does not violate physical laws must ensure positive density and positive values for the isothermal compressibility.
Thus, the domain for the pressure would then  be
\begin{flalign}\label{eq:domain-Pi}
    \domP^\text{phys} = 
    \{p \in \mathbb{R} ~|~ \rho=\Pi'(p) > 0,\ \Pi''(p) > 0\}. &
\end{flalign}
\begin{definition}\label{def:feasible}
  A \emph{feasible} solution to the $\GF$ problem  is  
  a solution of Eq. \eqref{eq:GF} that
  restricts pressures to be in $\domP^\text{phys}$ defined in 
  Eq.~\eqref{eq:domain-Pi} with  non-negative mass flow through compressors.
\end{definition}

For an ideal gas, as well as the CNGA EoS, Definition~\ref{def:feasible} rules that  $\domP^\text{phys} = \{p > 0\}$. 
Moreover, as a consequence of restricting the allowable values of the pressure to this domain, $\rho > 0$ ensures that $\Pi$ is strictly increasing on $ \domP^\text{phys}$ so that $\Pi(p) > \Pi(0)\ \forall \  p \in \domP^\text{phys}$. 
This is why both \cite{Singh2020, misra2020monotonicity}, in proving uniqueness of solution assume positive pressures and densities, respectively. 
The proof of uniqueness for an ideal gas in \cite{Singh2020} uses graph theoretic arguments in conjunction with the monotonicity of $\Pi$ for positive pressures. However, in the proof provided in \cite{misra2020monotonicity}, monotonicity features more prominently. 

In practice, when using a non-linear iterative scheme to solve the $\GF$ problem, 
the algorithm may yield an unphysical/infeasible solution, e.g. with negative pressures or $\Pi(\cdot) < \Pi(0)$, 
and one cannot conclude from the existing uniqueness results whether or not the problem has an alternative 
\emph{feasible} solution as in Definition~\ref{def:feasible}.
This motivates the study of a uniqueness result that pertains to mathematical (possibly non-physical) solutions of the $\GF$ problem by relaxing the conditions imposed on $\domP^\text{phys}$.

Before presenting our uniqueness results, we first introduce two definitions that separate the mathematical analysis of the system in Eq. \eqref{eq:GF} from its physical interpretation.
\begin{definition}
\label{def:compatible}
   Potentials $\Pi_i, \Pi_j$ are compatible with a compressor $(i,j)$ if there exist pressures  $p_i$ and $p_j$ that satisfy Eq. \eqref{eq:GFcompressor} with $\Pi(p_i) = \Pi_i, \Pi(p_j) = \Pi_j$. 
   A potential $\Pi_i$ is compatible with a slack node $i\in N_s$ if $\Pi_i = \Pi(p_i)$ for the specified slack pressure
   $p_i$ in Eq. \eqref{eq:GFslack}. 
 \end{definition}
 
\begin{definition}
\label{def:generalized}
   A \emph{generalized} solution to the $\GF$ problem Eq. \eqref{eq:GF} is a 
   set of values for the potential and mass flows that satisfies Eq. \eqref{eq:GFpipe} and Eq. \eqref{eq:GFbalance} 
   and is compatible with the compressors and slack nodes as per Definition~\ref{def:compatible}.
   Generalized solutions (unlike feasible solutions)
   have no constraints on the sign of pressure or the direction of mass flow in compressors.
 \end{definition}

If we drop the requirement that $\Pi''(\cdot) > 0$, then we may define a modified domain $\domP^\text{gen}$ for our purposes as follows:

 \begin{figure}[htb]
     \centering
     \includegraphics[scale=1.0]{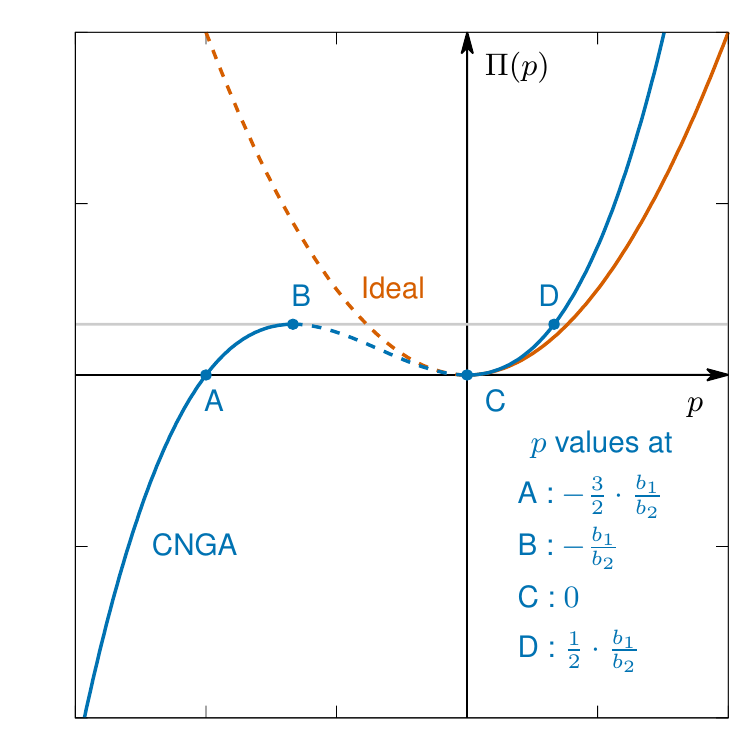}
     \caption{Potential function (not to scale) for the ideal and CNGA EoS, assuming $b_1, b_2$ are positive. The solid lines indicate the parts of the function that are increasing (have positive density) and serve as candidates for  the domain of definition according to Assumption~\ref{assumption:domain-Pi-relaxed}. For the ideal gas,  the domain can only be $\{p > 0\}$ because the function is decreasing otherwise. For the CNGA EoS, $\domP^\text{gen}$ can be taken as the abscissa of the curve excluding either the portion  A to C or the portion  B to D. 
     However, the feasible domain $\domP^\text{phys}$ as indicated by Eq. \eqref{eq:domain-Pi} would be the same for both the curves, i.e., the set $\{p > 0 \}$.}
     \label{fig:domain}
 \end{figure}

 
  \begin{assumption}
 \label{assumption:domain-Pi-relaxed}
There exists a non-empty set $\domP^\text{gen} \subseteq \mR$ such that if  $p,  \hat p\in \domP^\text{gen}$, then  $\Pi(p) > \Pi(\hat p)$ whenever $p > \hat p$. Such a set is maximal in the sense that inclusion of any $p \notin \domP^\text{gen}$ would invalidate the monotonicity property defined above.
 \end{assumption}
 
 
 %
Note that the set $\domP^\text{gen}$ need not be unique but it can be constructed to satisfy $\domP^\text{gen} \supset \domP^\text{phys}$. 
 For the ideal and CNGA EoS, such a domain that includes $\domP^\text{phys}$  is indicated in Fig. \ref{fig:domain}.  
 For the ideal gas, the set $\domP^\text{gen} = \{p > 0\}$ is the only possible choice  while for the CNGA EoS Eq. \eqref{eq:rho-nd} one could choose either  
 $\domP^\text{gen} = \left\{p \leq -\frac{3}{2}\cdot\frac{b_1}{b_2} \right\} \cup \{p > 0\}$ or instead 
 $\domP^\text{gen} = \left\{p \leq -\frac{b_1}{b_2} \right\} \cup \{p > \frac{1}{2}\cdot\frac{b_1}{b_2}\}$. However, the latter does not contain $\domP^\text{phys}$. In either case, a choice is made so that $\domP^\text{gen} $ is then fixed.
 %

%
 \begin{lemma}
\label{lemma:Pi-compressor-monotonic}
    A compressor preserves monotonicity of potentials under Assumption~\ref{assumption:domain-Pi-relaxed}, 
    i.e. if $\Pi_i < \hat \Pi_i$ are potentials at the inlet of a compressor $(i, j)$, 
    then the potentials $\Pi_j$ and $\hat \Pi_j$ at the discharge  satisfy $\Pi_j < \hat \Pi_j$.
 \end{lemma}
\begin{proof}
 There exist $p_i < \hat p_i$ with $\Pi(p_i) = \Pi_i, \Pi(\hat p_i) = \hat \Pi_i$.
 For a compressor boost $\alpha > 0$, $p_j = \alpha p_i < \alpha \hat p_i = \hat p_j$. Since $\Pi$ is increasing, $\Pi(p_j) < \Pi(\hat p_j)$.
\end{proof}
For the case of an arbitrary non-ideal gas considered here, we will prove that if a \emph{generalized} solution (Definition~\ref{def:generalized})  to the $\GF$ problem  Eq. \eqref{eq:GF} satisfies Assumption~\ref{assumption:domain-Pi-relaxed}, then it is unique. 
The implication of uniqueness of such a \emph{generalized} solution is that either it is the unique \emph{feasible} solution (Definition~\ref{def:feasible}) or the problem is infeasible.

Our proof closely follows that of \cite{Singh2019, Singh2020} who proved uniqueness of the solutions with positive pressure, i.e.,  \emph{feasible} solutions. Hence 
we shall state the essential  lemmas and theorems  in \cite{Singh2019, Singh2020} used to prove the uniqueness result and explain how they  hold for a \emph{generalized} solution under Assumption~\ref{assumption:domain-Pi-relaxed}.
One obvious feature of these proofs is that the  validity of the arguments does not hinge on non-negativity of the compressor flows. This observation helps to expand the proof of uniqueness to include a \emph{generalized} solution.

\begin{proposition} \label{prop:s-t-path}
(\cite{Korte2018}, Lemma~5 in \cite{Singh2020})
Given a graph with (balanced) injections and withdrawals, 
there exist nodes $m, n$ such that there is injection at $m$, withdrawal at $n$, and a non-intersecting path (i.e., a path where no vertex is repeated) from $m$ to $n$ 
with flow directions along the path.
\end{proposition}

\begin{proposition} 
\label{prop:flow_increase}
(Lemma~2 in \cite{Singh2020})
Consider a path between slack nodes $m, n$ along edges $\{1, 2, \dotsc k\}$. 
If flows $f_{1}, f_{2},\dotsc$ as well as $\hat f_{1}, \hat f_{2},\dotsc$ satisfy Eq. \eqref{eq:GFpipe}, \eqref{eq:GFbalance} following Assumption~\ref{assumption:domain-Pi-relaxed}, 
then it is not possible for all $i = 1, 2, \dots k$ to have  
$f_{i} > \hat f_{i}$.
Similarly it is not possible to have $f_{i} < \hat f_{i}$ for all $i = 1, 2, \dots k$.
\end{proposition}
\begin{remark}
Note that without \ref{assumption:path-pipe}, the proposition \ref{prop:flow_increase} is false.
\end{remark}

The proof in \cite{Singh2020}  uses the monotonicity of the pressure drop as a consequence of flow in pipes and compressors. Since Lemma~\ref{lemma:Pi-compressor-monotonic} proves that compressors preserve the monotonicity of potentials, the veracity of the proposition is verified. Armed with these results, we may now state the following lemma.
\begin{lemma}\label{lemma:injections_uniqueness}
If a \emph{generalized} solution to the $\GF$ problem Eq. \eqref{eq:GF}  exists under Assumption~\ref{assumption:domain-Pi-relaxed},
then the given tuple $(\bm q, \bm p)$, 
where $\bm q \in \mR^{|N_{ns}|}$, $\bm p \in \mathbb{R}^{|N_s|}$, 
determines $\bm q^{\mathrm{full}} \in \mathbb R^{|N|}$ subject to Eq. \eqref{eq:injection-extraction balance} uniquely.
In other words, the components of $\bm q^{\mathrm{full}}$ corresponding to the slack nodes are well-defined.
\end{lemma}
\begin{proof}
Let $\bm x \in \mathbb{R}^{|N_s|}$ denote the vector of unknown injections at the slack nodes.
We know from Eq. \eqref{eq:injection-extraction balance} that $$\underset{i\in N_s}{\sum} x_i + \underset{j\in N_{ns}}{\sum} q_j = 0.$$
Suppose $\bm x \in \mathbb{R}^{|N_s|}$ is not unique. Then  two distinct slack injections $\bm x$ and $\hat {\bm x} \in \mathbb{R}^{|N_s|}$ with $\bm x \neq \hat{\bm x}$ will result in distinct mass flows $\bm f, \hat{\bm f} \in \mathbb R^{(|P|+|C|)}$ with $\bm f \neq \hat {\bm f}$ that satisfy the nodal balance condition
$$\bm A^{\mathrm{full}}\bm f = \begin{bmatrix} \bm q \\ \bm x \end{bmatrix} \text{ and } \bm A^{\mathrm{full}}\hat {\bm f} = \begin{bmatrix} \bm q \\ \hat{\bm x} \end{bmatrix}. $$
Subtracting we get $\bm A^{\mathrm{full}}(\bm f - \hat{\bm f}) = \begin{bmatrix} 0 & \bm x - \hat{\bm x} \end{bmatrix}^T$. 
Note that $$\underset{i\in N_s}{\sum} (x_i - \hat{x}_i)  = 0.$$
By Proposition~\ref{prop:s-t-path}, there exists a path from a node $m$ with $x_m - \hat{x}_m > 0$ to a node $n$ with $x_n - \hat x_n < 0$ along which  flow satisfies $f_l > \hat f_l$ for every edge $l$. 
But Proposition~\ref{prop:flow_increase} contradicts this since pressures at the start and the end ($m$ and $n$) are fixed. 
Thus, we must have $\bm x = \hat{\bm x}$ and $\bm q^{\mathrm{full}}$ is unique.
\end{proof}

\begin{theorem}\label{theorem:graph_uniqueness}

If a \emph{generalized} solution (Definition~\ref{def:generalized})  to the $\GF$ problem  Eq. \eqref{eq:GF} satisfying Assumption~\ref{assumption:domain-Pi-relaxed} exists, then it is unique.  
\end{theorem}

\begin{proof}
In \cite{Singh2020}, the uniqueness of a \emph{feasible} solution for $|N_s| > 1$ is proved by first  establishing the uniqueness of $\bm q^{\mathrm{full}}$ so that the problem is tantamount to proving uniqueness for the case $|N_s| = 1$, as was done in \cite{Singh2019}. We have established that $\bm q^{\mathrm{full}}$ is unique, 
so as in the proof (Theorem~1) in \cite{Singh2019} that  uses the structure of the nullspace of $\bm A^{\mathrm{full}}$ along with an analogue of Proposition~\ref{prop:flow_increase} stated for a cyclic path, the same logical arguments carry through here to obtain uniqueness of a \emph{generalized} solution that satisfies Assumption~\ref{assumption:domain-Pi-relaxed}.
\end{proof}
\begin{remark}
Note that the proof above relies on the validity of \ref{assumption:cycle-pipe}.
\end{remark}

\begin{remark}
Assumption~\ref{assumption:domain-Pi-relaxed} needs to be checked only for nodes with compressors in Eq.~\eqref{eq:GFcompressor}, because Eq.~\eqref{eq:GFpipe} involves only the potential, and pressure is always known in Eq.~\eqref{eq:GFslack}.
\end{remark}
\begin{remark}
The results so far have been derived  considering compressors as edge elements of the network that provide multiplicative boost $\alpha \geqslant 1$ in Equation~\eqref{eq:GFcompressor}. However, the results hold as long as $\alpha > 0$, so that that this theory carries over to networks that consist of  components such as pressure regulators, short pipes, valves etc, all of which are modelled by Equation~\eqref{eq:GFcompressor} with $\alpha > 0$.
\end{remark}
\begin{remark}
\label{remark:potential_reference}
Note that the results derived thus far are valid for any choice $p_0 \in \mR$ in Eq.~\eqref{eq:pi-def}. However, for convenience, hereon we shall assume  $p_0 = 0$ so that $\Pi(0) = 0$.
\end{remark}
\subsection{Feasibility for flow in a single pipe}
\label{subsec:single_pipe}

While it is not possible to analytically conclude if the problem \eqref{eq:GF} is feasible for a network, 
one can motivate the conditions that lead to infeasibility by considering flow in a single pipe. In other words, we shall examine conditions under which a solution exists for the system \eqref{eq:GF} governing flow in a single pipe.  
Consider a single pipe (see Fig. \ref{fig:single-pipe}) where at one end we prescribe the flow and at the other end we prescribe pressure. 
The mass flow is  determined from a trivial application of the nodal balance equation.
Thus, for concreteness, consider  a pipe  $(1, 2)$ with given pressure $p_1 >0$ (or $p_2 > 0$) and flow $f_{12}>0$ from $1\rightarrow 2$.

If $p_2 > 0$ were known, then of course the corresponding equation for $p_1$, 
$\Pi(p_1) = \Pi(p_2) +  \beta_{12}f_{12}^2 > 0$  will always be feasible.
The question of feasibility for given $p_1 > 0$ is tantamount to asking if the equation $\Pi(p_2) = \Pi(p_1) -  \beta_{12}f_{12}^2$ has a solution $p_2 > 0$.
Clearly, if $\Pi(p_1) -  \beta_{12}f_{12}^2 < 0$, the equation has no solution since $\Pi(\cdot) > 0$ for positive arguments. However, if $\Pi(p_1) -  \beta_{12}f_{12}^2 > 0$, then there is a unique solution $p_2 > 0$ since  $\Pi$ is an increasing function for positive arguments. Thus infeasibility results when  the slack pressure ($p_1$) fails to sustain the given mass flow $f_{12}$  for any pressure $p_2 > 0$. 

\begin{figure}
    \centering
    \includegraphics[scale=0.8]{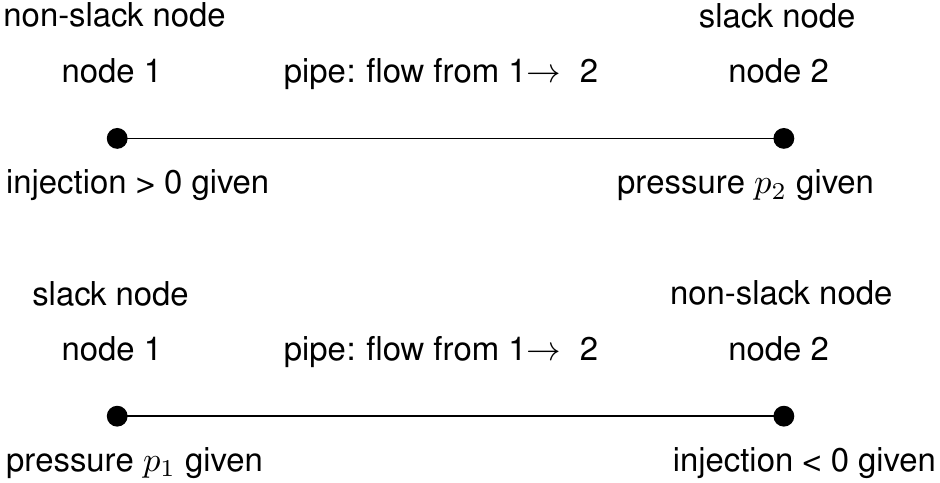}
    \caption{Illustration of the two possibilities for a single pipe instance considered in Sec. \ref{subsec:single_pipe}. In the first case (top), the problem is always feasible, but in the second case (bottom) infeasibility may occur.}
    \label{fig:single-pipe}
\end{figure}

Thus, we have demonstrated that infeasibility can occur (solution may not exist) even in the case of a single pipe by analysing one non-linear equation, but this analysis cannot be performed for a large network since we will then have a system of  non-linear equations that need to be considered simultaneously.
We now present a brief overview of the algorithm used to solve the $\GF$ problem.

\section{Algorithms} \label{sec:algo}
The algorithm we use to solve Eq. \eqref{eq:GF} is a classical \NR algorithm without line search.
Given the $k$\textsuperscript{th} iterate pressures $p_i^{(k)}$ and mass flows $\bm f^{(k)}$ for $k= 0, 1, 2, \dotsc$, the \NR iteration computes the $(k+1)$\textsuperscript{th} iterate from the increments  $\Delta p_i^{(k)}$ and $\bm \Delta \bm f^{(k)}$ as
\begin{subequations}
\begin{align}
& p_i^{(k+1)} = p_i^{(k)} + \Delta p_i^{(k)} \,\, &\; \forall \; i \in N,\\
& f_{ij}^{(k+1)} = f_{ij}^{(k)} + \Delta f_{ij}^{(k)} \,\, &\; \forall \; (i, j) \in P \cup C. 
\end{align}
\end{subequations}
The increments are determined from the residuals
\begin{subequations}
\begin{align}
& r_{\mathrm{pipe}}^{(k)} = \Pi \left(p_i^{(k)}\right) - \Pi\left(p_j^{(k)}\right) - \beta_{ij}\lvert f_{ij}^{(k)}\rvert f_{ij}^{(k)} \;  \in \mathbb{R}, \\
& r_{\mathrm{comp}}^{(k)} = p_j^{(k)} - \alpha_{ij} p_i^{(k)} \; \in \mathbb{R}, \\
& \bm r_{\mathrm{node}}^{(k)} = \bm A  \bm f^{(k)} - \bm q \;  \in \mathbb R^{|N_{ns}|}, 
\end{align}
\end{subequations}
as well as the linear system defined by the Jacobian:
\begin{subequations}
\begin{align}
& \Pi'\left(p_i^{(k)}\right)\Delta p_i^{(k)} - \Pi'\left(p_j^{(k)}\right)\Delta p_j^{(k)}
\notag \\
&\quad - 2\beta_{ij}\lvert f_{ij}^{(k)} \rvert \Delta f_{ij}^{(k)} = - r_{\mathrm{pipe}}^{(k)} \quad & \forall (i, j) \in P, \\ 
& \Delta p_j^{(k)} - \alpha_{ij} \Delta p_i^{(k)} = - r_{\mathrm{comp}}^{(k)}  \quad & \forall (i, j) \in C, \\ 
& \bm A  \bm \Delta \bm f^{(k)}   = - \bm r_{\mathrm{node}}^{(k)}, \\
& \Delta p_i^{(k)} = 0 \; & \forall i \; \in N_s.
\end{align}
\label{eq:NR_system}
\end{subequations}

In Appendix~\ref{appendix}, we state and prove Lemma~\ref{lemma:Jacobian_injective} which yields sufficient conditions for the invertibility of the Jacobian and allows us to make a suitable initial guess.
We set termination conditions for the algorithm by specifying  a small tolerance for the residual below which the solver is said to have converged to a solution. If the required tolerance is not achieved when $k = 2000$ iterations are complete, we declare failure and terminate the iteration.

It is well-known that it is not possible to determine \emph{a priori} if and when a \NR algorithm will converge to a solution for a general non-linear system of equations such as \eqref{eq:GF}. 
The solvability of a system is usually determined by looking for a solution using a variety of methods, and if one is found, invoking  uniqueness to conclude that it must be \emph{the} solution. Thus, when the \NR algorithm is used to solve the system \eqref{eq:GF}, exactly one of the following events will occur:

\begin{enumerate}[label=(E\arabic*)]
    \item  The iterates converge to a \emph{generalized} solution that satisfies Assumption~\ref{assumption:domain-Pi-relaxed}. \label{list:unique}
    \item  The iterates converge to a \emph{generalized} solution that \emph{does not} satisfy Assumption~\ref{assumption:domain-Pi-relaxed}. \label{list:converged}
    \item  The iterates do not converge to a solution. \label{list:not_converged}
\end{enumerate}

If \ref{list:unique} occurs,  uniqueness (Theorem~\ref{theorem:graph_uniqueness}) provides definitive  closure to the problem, for either the solution is \emph{feasible} (Definition~\ref{def:feasible}) or it is un-physical, and thus determines the problem to be infeasible.  Nothing conclusive may be said in case of  \ref{list:converged} and \ref{list:not_converged} for they could represent either infeasibility or failure of the algorithm to converge even though a solution exists. 

\subsection{\NR for Ideal Gas EoS}
\label{subsec:pressure-correction-ideal}
Suppose that the classical \NR algorithm is applied to the $\GF$ problem in Eq. \eqref{eq:GF} for an ideal gas ($b_2 = 0$ in Eq. \eqref{eq:rho-nd}), and results in a feasible solution.  The form of the potential function (Fig.~\ref{fig:domain}) ensures that the positive pressure solution always exists and is compatible with the compressors at all nodes, so that \ref{list:converged} cannot occur. 
If there is negative mass flow through any compressor, then the solution is un-physical and the problem is infeasible.


\subsection{\NR for Non-ideal Gas (CNGA EoS)}
\label{subsec:pressure-correction-cnga}

If the classical \NR algorithm applied to the $\GF$ problem for a non-ideal gas ($b_2 > 0$ in Eq. \eqref{eq:rho-nd}) 
results in a solution, both \ref{list:unique} and \ref{list:converged} are possible.
The \emph{generalized} solution is evaluated to determine whether it satisfies Assumption~\ref{assumption:domain-Pi-relaxed}, i.e., \ref{list:unique}. 
In case of \ref{list:unique}, any node with a negative potential or any compressor with a negative mass flow implies infeasibility. For nodes with positive potentials, we know from Fig. \ref{fig:domain} that unique positive pressures exist which can be computed either from  pipe equations Eq. \eqref{eq:GFpipe} or calculated systematically in a second run of the \NR algorithm as described below.

Suppose the algorithm terminates with negative pressures.
In order to initiate another run of \NR,  we start with an initial guess that perturbs the relevant solution components to positive values, say $|p_{n_1}^{*}|, |p_{n_2}^{*}|, \dotsc$ while keeping all the other solution components unchanged.
Note that residuals which do not depend on these nodal pressures will remain zero. 

\section{Results} \label{sec:results}
We now present the results of several computational experiments that corroborate the effectiveness of the proposed algorithm in solving the steady-state gas flow problem for both ideal and non-ideal gas. Henceforth, we shall use the term non-ideal gas to mean one that is described by the CNGA EoS. The algorithms presented in Section \ref{sec:algo} are all implemented using the Julia programming language \cite{Bezanson2017} and the code is released as an open-source Julia package: \url{https://github.com/kaarthiksundar/GasSteadySim.jl}. Furthermore, all the computational experiments were run on a MacBook Pro with 2.8 GHz Quad-Core Intel Core i7 processor and 16 GB of RAM. The code to run all the experiments presented in this article and generate the plots and tables in the subsequent paragraphs can be found at the website: 
\url{https://github.com/kaarthiksundar/GasFlowRuns}.

\subsection{Description of the test cases} \label{subsec:test-cases}
The first test case is that of a single pipeline. We consider a single, 36-inch diameter pipe with  friction factor $\lambda = 0.01$ that connects two nodes $1$ and $2$. Node $1$ is a slack node with a slack pressure of \SI{4.3}{\mega\pascal}. Node $2$ is a non-slack node with a gas withdrawal of \SI{275}{\kilogram\per\second}. 
The second set of test cases comprises five natural gas networks taken from GasLib \cite{Schmidt2017}. The five instances are GasLib-11, GasLib-24, GasLib-40, GasLib-134 and GasLib-582 with  11, 24, 40, 134 and 582 nodes in the network respectively. The values for the pipe diameter, friction factor, nodal injections etc. for each of these networks can be obtained at \url{https://gaslib.zib.de/} and \url{https://github.com/kaarthiksundar/GasFlowRuns}. The GasLib instances, apart from having pipes and compressors, may also contain some additional physical components such as valves and regulators as well as some non-physical components such as resistors and  loss-resistors. In our computational experiments, the non-physical components are treated as pass-through elements, i.e., they are modelled to allow any amount of gas to flow through them without incurring a pressure drop. All the valves are assumed to be closed, and the regulators (pressure-reducing components) are assumed to not allow any reduction in pressure. Each GasLib network specifies a base amount of injection or withdrawal in each node that satisfies global network balance. 

For each GasLib network, 500 instances are generated according to the following procedure: The node with the largest injection in the network is chosen to be the slack node where the slack pressure is fixed to \SI{5}{\mega\pascal}. The gas injection in each non-slack node is then scaled using a random uniform factor in the range \numrange{0.9}{1.1}. Finally, the compressor ratio is set to a random value in the range \numrange{1.1}{1.4} for each compressor and the pressure-reduction ratio for each regulator is set to 1. This procedure is repeated to generate 500 instances for each network.  

\subsection{Importance of non-dimensionalization} \label{subsec:nd-value}
This set of results is aimed at showing the influence of non-dimensionalization  in the convergence of the \NR algorithm  for the steady-state gas flow problem. To that end, we consider the gas flow problem on the GasLib networks and solve both the dimensional and non-dimensional versions of the problem using the algorithm in Section \ref{sec:algo} recognising that the dimensional version of the problem is equivalent to setting the nominal values $\rho_0 = p_0 = v_0 = A_0 = 1$ instead of Eq. \eqref{eq:nominal-quantities}. The \NR  algorithm was successful in computing the gas flow solution in \emph{every instance when the equations were non-dimensionalized}. As for the dimensional version of the steady-state equations, Table \ref{tab:dim-vs-non-dim} shows the number of instances (out of 500) in which the gas flow problem was solved successfully.  The table shows the value of non-dimensionalizing the equations in ensuring convergence with the standard \NR algorithm when applied to the gas flow problem for both ideal and non-ideal gases. 
\begin{table}[!htb]
    \centering
    \caption{Number of instances (out of 500) when the \NR algorithm solves the $\GF$ problem for ideal and non-ideal gas using the dimensional form of the steady-state equations. The algorithm solves all 500 instances when non-dimensionalizing.}
    \label{tab:dim-vs-non-dim}
    \begin{tabular}{ccc}
        \toprule
        instance & ideal & non-ideal \\
        \midrule
        \begin{filecontents*}{tables/d-vs-nd-counts.csv}
instance,ideal_d,ideal_nd,nonideal_d,nonideal_nd
GasLib-11,183,500,82,500
GasLib-24,104,500,64,500
GasLib-40,23,491,3,491
GasLib-134,145,500,113,500
GasLib-582,0,343,0,341
\end{filecontents*}
\csvreader[late after line=\\]{tables/d-vs-nd-counts.csv}{1=\one,2=\two,3=\three,4=\four,5=\five}{\one & \two & \four}
        \bottomrule
    \end{tabular}
\end{table}

While these results demonstrate the importance of non-dimensionalization by considering  non-trivial networks, it can also be seen in the case of a trivial network, namely, the case of a single pipe as discussed in Section~\ref{subsec:single_pipe}. 
Recall that we considered the equation
\begin{equation}
    \Pi(p_2) = \Pi(p_1) -  \beta_{12}f_{12}^2 
\end{equation}
for given positive values of $p_1, f_{12}, \beta_{12}$.
If the equation is feasible, we argued that it has a unique positive solution $p_2$. For a non-ideal gas described by the CNGA EoS, the expression for $\Pi$ corresponds to a cubic polynomial and the positive root needs to be approximated numerically. We found that for the non-dimensional version of the equation, an appeal to standard root-finding routines  led to accurate and expected results, but the same routines produced  spurious solutions when applied to the dimensional version of the problem.

\subsection{Influence of EoS on Flow Behaviour} \label{subsec:ideal-vs-non-ideal}
While it is convenient to assume ideal gas behaviour for flow in pipeline networks, the assumption deviates from reality; the ideal gas EoS is not appropriate for natural gas in transmission pipeline networks. We compare the flow behaviour in different test cases assuming the ideal gas EoS as well as the CNGA EoS, and show the importance of taking this into account when solving the gas flow problem.
\begin{figure}[!htb]
    \centering
    \includegraphics[scale=1.0]{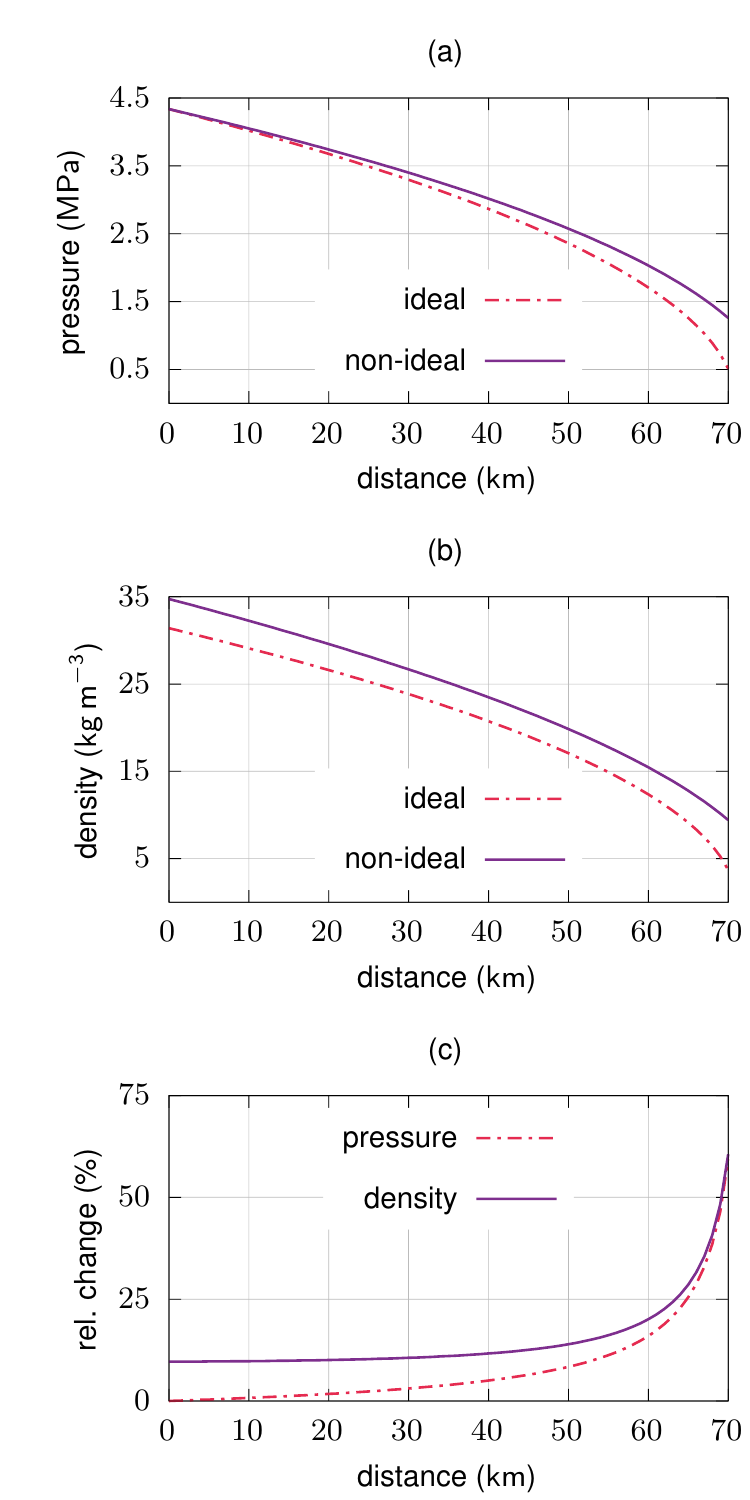}
    \caption{Pressure and density solutions obtained by solving the gas flow problem for a single pipe with an ideal gas EoS and CNGA EoS respectively. The Figure (c) shows the percentage change between the pressure and density solutions for ideal and non-ideal gases, relative to the non-ideal gas. }
    \label{fig:ideal-vs-nonideal}
\end{figure} 
To that end, we first consider the test case of a  single pipe and solve the gas flow problem with both ideal and CNGA  EoS. Fig. \ref{fig:ideal-vs-nonideal} assimilates the results for the problem when the pipe length is \SI{70}{\km}. Fig. \ref{fig:ideal-vs-nonideal} (a) and (b) show the pressure and density solutions for both the ideal and non-ideal EoS respectively. Fig. \ref{fig:ideal-vs-nonideal} (c) shows the percentage change in the pressure and density solutions relative to the case of the non-ideal gas. From this figure, we see that as the length of the pipe increases, the relative percent change  can increase to as much as \SI{50}{\percent}. This illustrates that use of the ideal gas EoS can underestimate the pressure and density drop to a substantial extent, thus emphasising the need to use non-ideal gas EoS when solving the gas flow problem. 

The box plots in Fig. \ref{fig:pressure-max} and  Fig. \ref{fig:density-max} show the statistics of maximum relative difference in nodal pressures and densities for the GasLib networks when a CNGA EoS is used. Unlike the case of the single pipe where the deviation in the pressure solution was large, here the maximum relative difference over all the nodes of the system is fairly small. In this context we remark that this deviation is a function of multiple factors like pipe length, friction factor and amount of flow passing through the pipeline and hence, it is very difficult to ascertain a-priori whether  an ideal gas EoS  will be sufficient for a given network.

\begin{figure}[!htb]
    \centering
    \includegraphics[scale=1.0]{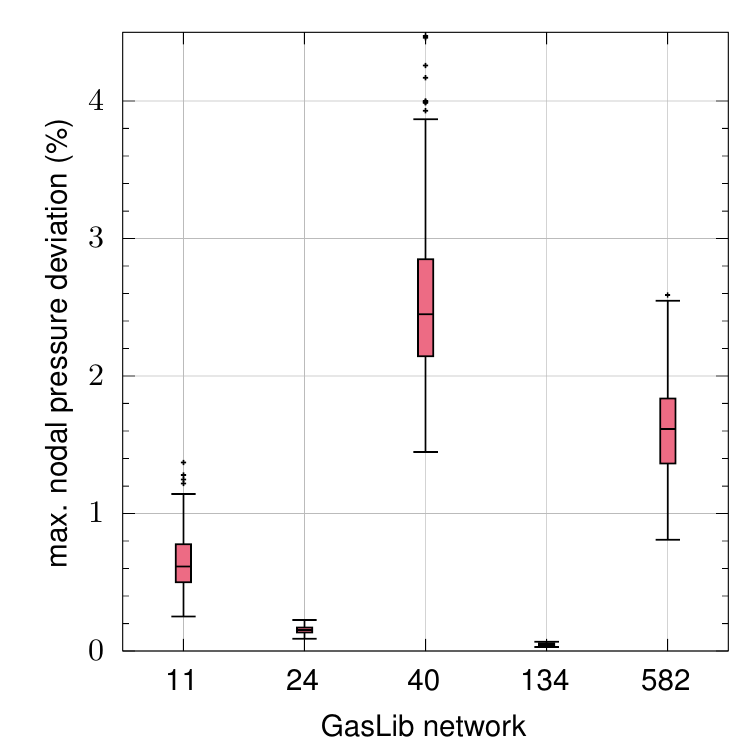}
    \caption{Box plot of the maximum relative nodal pressure deviation between the solution for ideal and non-ideal gases, over all the nodes.}
    \label{fig:pressure-max}
\end{figure}

\begin{figure}[!htb]
    \centering
    \includegraphics[scale=1.0]{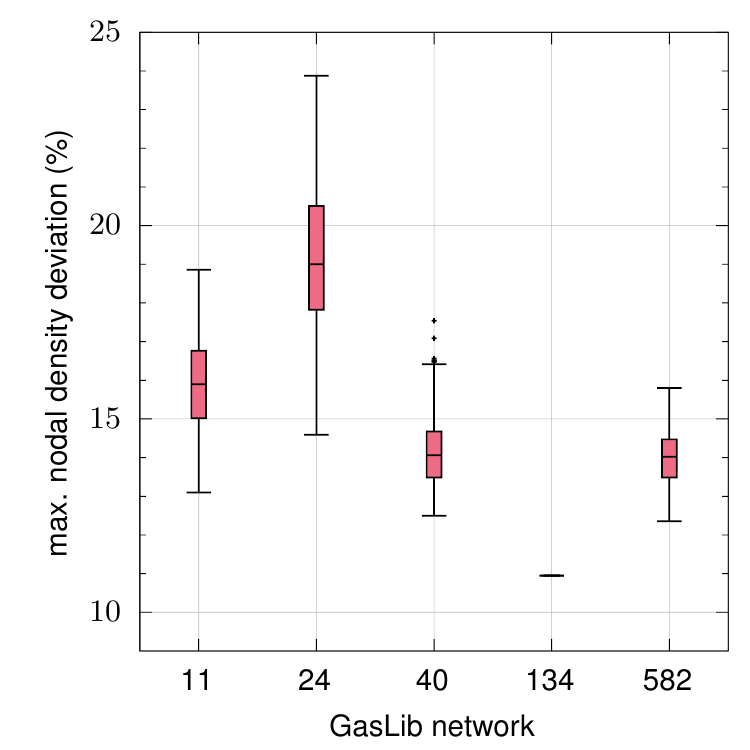}
    \caption{Box plot of the maximum relative nodal density deviation between the  solution for ideal and non-ideal gases, over all the nodes.}
    \label{fig:density-max}
\end{figure}

\subsection{Infeasibility detection} \label{subsec:infeasible} 
The discussions in Section~\ref{sec:algo} make it clear that whenever \ref{list:unique} occurs for the ideal or CNGA EoS,  negative potential at a node or negative mass flow in a compressor implies the system is infeasible. However, in  \ref{list:converged}, there is no way to conclude whether or not the problem is infeasible.

No cases of infeasibility or non-convergence occurred in any of the GasLib networks when perturbing the base withdrawals randomly in the range \numrange{0.9}{1.1} as described in Sec. \ref{subsec:test-cases}.
Hence, to present  computational results that illustrate the algorithm's ability to detect infeasibility arising in the data, we generate 500 additional instances of each GasLib network by perturbing the base withdrawals randomly in the range \numrange{0.75}{1.25}  as described in \cite{Singh2020}. 

All 500 instances for GasLib-11, GasLib-24, and GasLib-134 and were found to be feasible for both ideal and non-ideal gases. However, for both GasLib-40 and GasLib-582 instances, when the \NR algorithm was run for both an ideal and non-ideal gas, the algorithm detected 8 and 103 infeasible instances, respectively, each due to  negative flow in a compressor. The motivating study \cite{Singh2020} observed an almost identical number of infeasible cases for GasLib-40 in their  experiments on the $\GF$ problem with an ideal gas. The infeasibility detection approach in this article can be used for non-ideal gas modeling, and furthermore can localize the source of infeasibility and identify the compressor(s) or nodes causing it. 
The previous study \cite{Singh2020} presented a \NR solver that did not converge for any of the 500 instances unless a suitable initial guess was provided  by solution of a relaxed version of the equations.  Our technique enables convergence every time, even with a random initial guess.

\subsection{Multiple Slack Nodes}
A variant of the previous study is used to demonstrate that the solver can handle multiple slack nodes without issue. 
Examining the solution of the GasLib-40 base case with a single slack node (\#38  with slack pressure \SI{5}{\mega\pascal}), we designate additional slack nodes \#20 and \#40 with pressures \SI{7.41}{\mega\pascal} and \SI{2.80}{\mega\pascal} respectively. 
Then 500 instances are generated  by scaling the gas injection in each non-slack node using a random uniform factor in the range \numrange{0.75}{1.25} as before, and each compressor ratio is set to a random value in the range \numrange{1.45}{1.55}.
For this set of 500 runs, the \NR solver converged to a solution in \emph{every} instance for both ideal and non-ideal gas. For the ideal gas,  185 instances were infeasible because of negative flow in a compressor, while for the non-ideal gas the corresponding number was 174.
Note that because of the presence of multiple slack nodes, even slight variations in compressor ratios can make the problem infeasible. There is a dramatic increase in instances of infeasibility despite varying the compression ratio in a narrow range around the base case. We also did not observe a considerable change in the computation time of the algorithm to solve the GasLib-40 multiple slack runs in comparison to the single slack runs.

\subsection{Computational performance of the algorithm} \label{subsec:computation}
This set of results aims to summarize the statistics related to the computational performance of the algorithm when applied to the GasLib networks. 

\begin{table}[!htb]
    \centering
    \caption{Average number of iterations required by the \NR algorithm to solve the $\GF$ problem for ideal and CNGA EoS.}
    \label{tab:iterations}
    \begin{tabular}{ccc}
        \toprule
         instance & ideal & non-ideal \\
         \midrule
         \begin{filecontents*}{tables/iterations.csv}
instance,ideal,nonideal
GasLib-11,4,5
GasLib-24,3,5
GasLib-40,6,5
GasLib-134,2,5
GasLib-582,14,14
\end{filecontents*}
\csvreader[late after line=\\]{tables/iterations.csv}{1=\instance,2=\ideal,3=\nonideal}{\instance & \ideal & \nonideal}
         \bottomrule
    \end{tabular}
\end{table}

\begin{figure}[!htb]
    \centering
    \includegraphics{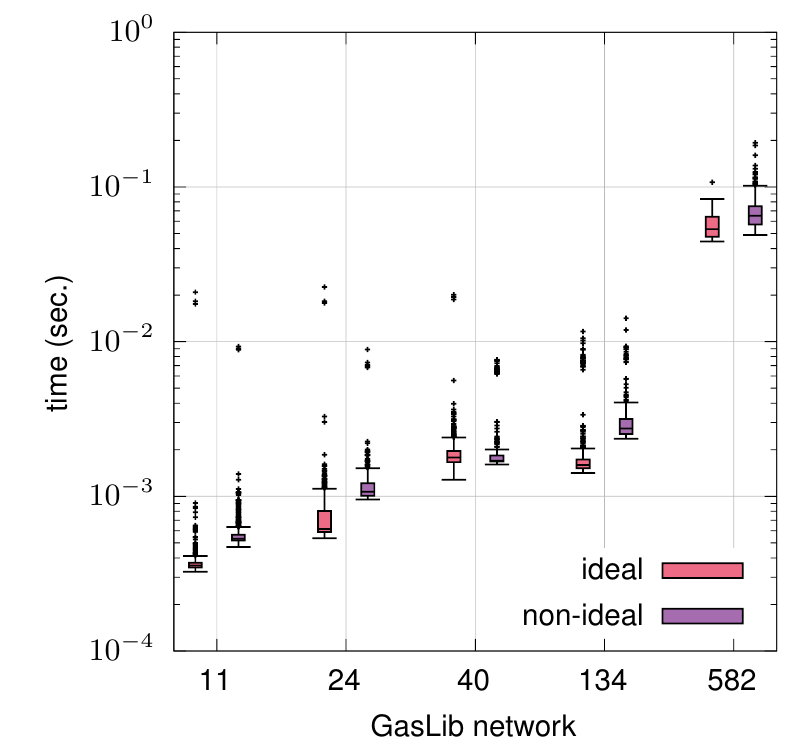}
    \caption{Box plot of the computation times for the 500 instances for every GasLib network considered in this article.}
    \label{fig:time}
\end{figure}

The Table \ref{tab:iterations} shows the average number of iterations taken by the \NR  algorithm to converge to the solution of the gas flow problem.  As can be observed from Table \ref{tab:iterations}, the number of iterations required rarely exceeds 20 for all the GasLib networks. This is important because when considering a  large-scale network, each iteration is computationally expensive and hence the algorithm could become impractical to use if it involves a large number of iterations even for simple networks. 

Fig. \ref{fig:time} shows the statistics of computation time taken by the algorithm to solve the gas flow problem for the 500 instances of each GasLib network considered in this article. Fig. \ref{fig:time} demonstrates the effectiveness of the algorithms in being able to compute a solution to the gas flow problem within a second of computational time. 

\section{Conclusion} \label{sec:conclusion}
We have presented the  steady-state network flow equations for a non-ideal gas and proved the uniqueness of mathematical solutions to the non-linear gas flow system under certain assumptions. 
Specifically, for an ideal gas and the non-ideal CNGA equation of state, we have outlined definite conditions under which solutions obtained by a non-linear solver may be acceptable or may determine infeasibility of the problem. 
Previous studies have asserted that the system of equations cannot be solved in practice with a classical \NR method without a carefully constructed initial guess, even for an ideal gas. 
However, we show that in all test cases considered, our  implementation of the \NR  algorithm converges in very few iterations independently of the initial guess. 
However, the \NR method is not an algorithm for which convergence can ever be guaranteed when the solution is unknown. Even the celebrated theorem of  Kantorovich \cite{Ciarlet2012Jul} that provides sufficiency conditions for the \NR iteration, does so in a manner that is not useful in practice. As such, one must be mindful of the fact that all guidelines offered can only be heuristic and empirical. However, our non-dimensionalization scheme is based on general thumb rules of numerical analysis - (i) That scaling the problem variables can lead to better conditioning which always aids the convergence of \NR (ii) while some problems are inherently ill-conditioned, bad scaling can turn even a well-conditioned problem into an ill-conditioned one.  Following these ideas, our non-dimensional scheme is one instance that looks to scale the variables using given data so that they are all of similar order of magnitude. While we have tested our scheme for multiple test problems, the general principle holds for any nonlinear system of equations.
We propose that the resulting improvement in conditioning of the entire system of equations leads to this desired performance.  
Finally, extensive computational experiments corroborate the effectiveness of the algorithms on benchmark instances, as well as the ability of the algorithm to identify and localize infeasibility of given data for a pipeline network.
Future work would focus on extending the flow solution problem to operational optimization in natural gas networks while accounting for non-ideal gas behaviour. 

\appendices
\section{Invertibility of the Jacobian}
\label{appendix}
Consider the following homogeneous linear system in $(x_k, y_{ij})$, where $k \in N,  \; (i,j) \in P \cup C$ defined for a gas network where the real numbers $\mathcal{B}_{ij} \neq 0, \; \mathcal{A}_i \neq 0, \; \alpha_{ij} > 0$ and all $\mathcal{A}_i$ have the same sign:
\begin{subequations}
\begin{align}
& \mathcal{A}_i x_i - \mathcal{A}_j x_j 
- \mathcal{B}_{ij} y_{ij} = 0 \quad \forall (i, j) \in P, \\ 
& x_j - \alpha_{ij} x_i = 0  \quad  \forall (i, j) \in C, \\ 
& \bm A  \bm y   = 0, \\
& x_i = 0 \quad  \forall i \; \in N_s,
\end{align}
\label{eq:Jacobian_injective}
\end{subequations}
\begin{lemma}
\label{lemma:Jacobian_injective}
The only solution to the homogeneous linear system defined by Equation~\eqref{eq:Jacobian_injective} is the trivial solution.
\end{lemma}
\begin{proof}
It is clear that $x_i = y_{ij} = 0$ is a solution.
If we identify the system as a $\GF$ problem governing ($x_i$, $y_{ij}$), then we see that the hypotheses of the uniqueness theorem (Theorem~\ref{theorem:graph_uniqueness}) hold, implying that there are no other solutions.
\end{proof}
In order to use the \NR iteration scheme, the Jacobian at the starting point must be invertible.  
In light of Lemma~\ref{lemma:Jacobian_injective}, if we now consider the linear system \eqref{eq:NR_system} obtained at an initial point,  the invertibility of the Jacobian is assured if  $$f_{ij}^{(0)} \neq 0 \; \forall \; (i, j) \in P, \quad \Pi'\left(p_i^{(0)}\right) > 0 \; \forall \; i \in N.$$
\begin{remark}
Note that if \ref{assumption:cycle-pipe} does not hold, then Equation~\eqref{eq:Jacobian_injective} has non-trivial solutions and the Jacobian becomes singular.
\end{remark}
\printbibliography
\vskip 0pt plus -1fil
\begin{IEEEbiography}[{\includegraphics[width=1in,height=1.25in,clip,keepaspectratio]{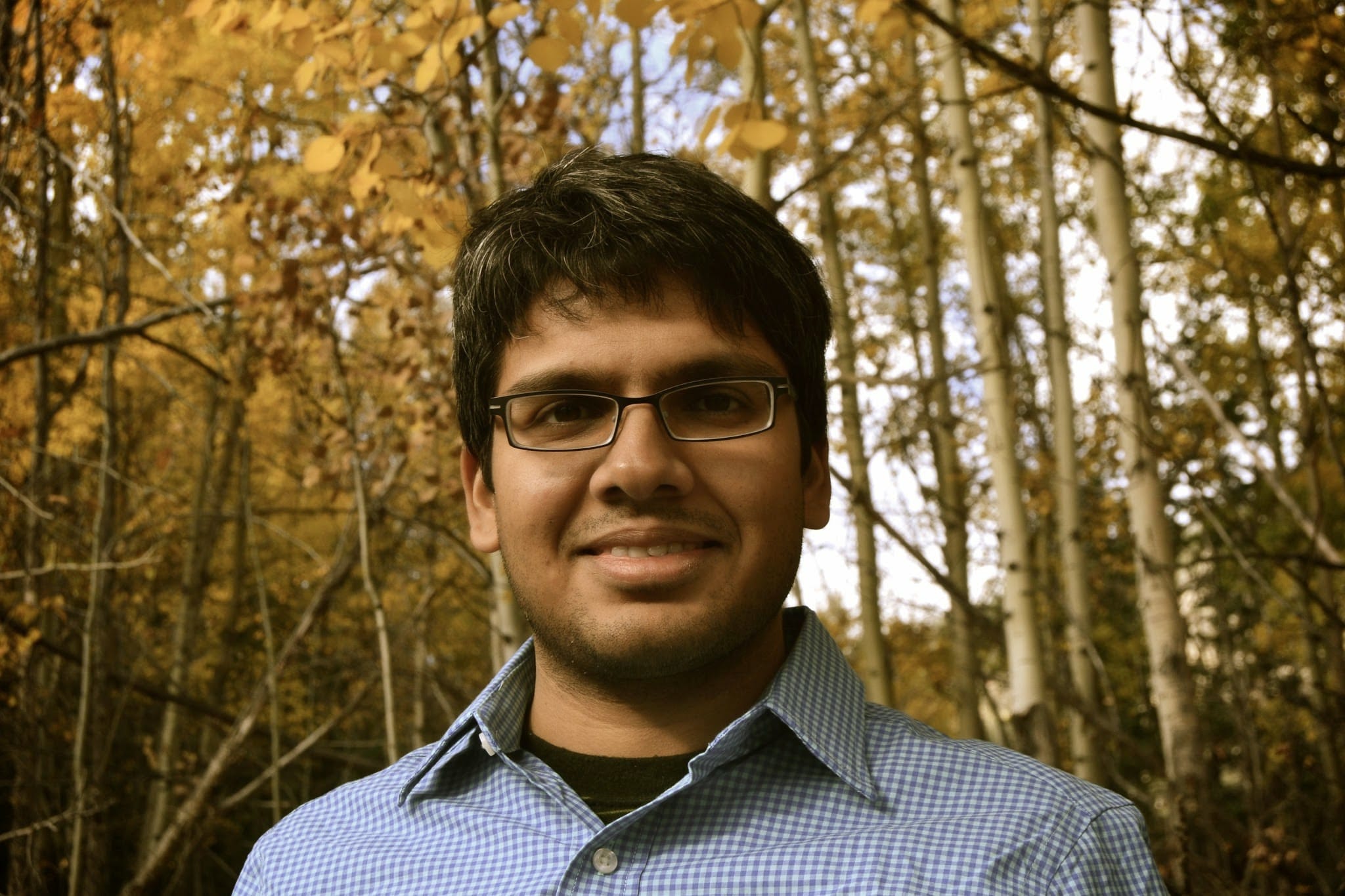}}]
{Shriram Srinivasan}
earned his M.S. in Mathematics and Ph.D. in Mechanical Engineering, all from Texas A\&M University, College Station. His research interests are in computational mechanics and reduced-order models of structured systems such as fracture networks and gas pipeline networks. He is a staff scientist in the Applied Mathematics and Plasma Physics Group of the Theoretical Division at Los Alamos National Laboratory, New Mexico.
\end{IEEEbiography}
\vskip 0pt plus -1fil
\begin{IEEEbiography}[
{\includegraphics[width=1in,height=1.25in,clip,keepaspectratio]{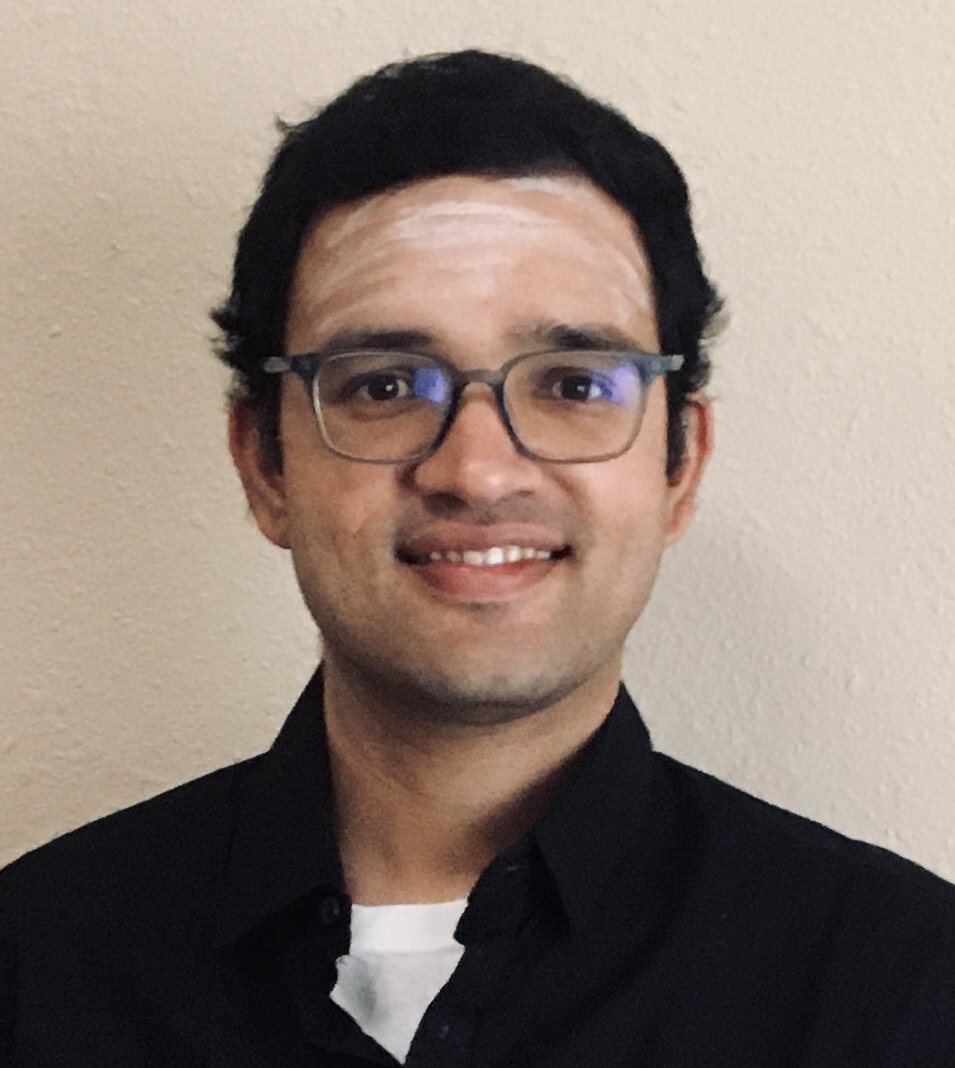}}
]{Kaarthik Sundar}
received the Ph.D. degree in mechanical engineering from Texas A\&M University, College Station, TX, USA, in 2016. He is currently a Research Scientist in the Information Systems and Modeling Group of the Analytics, Intelligence and Technology division at Los Alamos National Laboratory, Los Alamos, NM, USA. His research interests include problems pertaining to vehicle routing, path planning, and control for unmanned/autonomous systems; non-linear optimal control, estimation, and large-scale optimization problems in power and gas networks; combinatorial optimization; and global optimization for mixed-integer non-linear programs.
\end{IEEEbiography}
\vskip 0pt plus -1fil
\begin{IEEEbiography}[
{\includegraphics[width=1in,height=1.25in,clip,keepaspectratio]{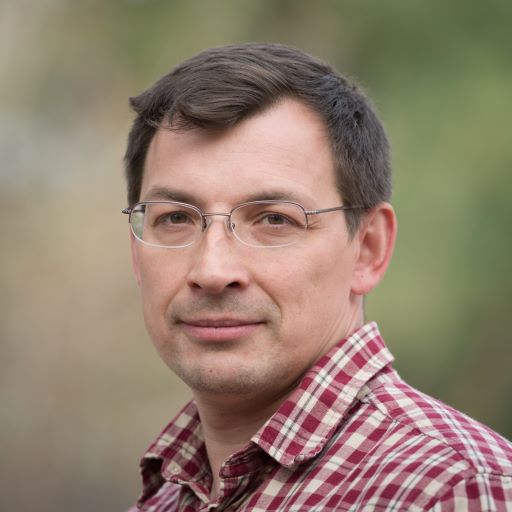}}
]{Vitaliy Gyrya} is a staff scientist in the Applied Mathematics and Plasma Physics group of the Theoretical Division at Los Alamos National Laboratory, where he was previously a postdoctoral associate at the Center for Nonlinear Studies.  
Before joining LANL in 2010, he obtained a Ph.D. in Applied Mathematics from The Pennsylvania State University, State College.
His research interest include compatible numerical discretizations and fluid flow problems.
\end{IEEEbiography}
\vskip 0pt plus -1fil
\begin{IEEEbiography}[
{\includegraphics[width=1in,height=1.25in,clip,keepaspectratio]{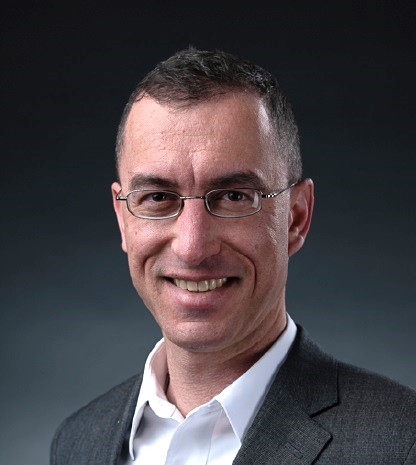}}
]{Anatoly Zlotnik} is a staff scientist in the Applied Mathematics and Plasma Physics group of the Theoretical Division at Los Alamos National Laboratory, where he was previously a postdoctoral associate at the Center for Nonlinear Studies.  Before joining LANL in 2014, he obtained a Ph.D. in systems science and mathematics from Washington University in St. Louis, Missouri, an M.S. in applied mathematics from the University of Nebraska – Lincoln, and B.S. and M.S. degrees in systems and control engineering from Case Western Reserve University in Cleveland, Ohio.  His research focus is on computational methods for optimal control of large-scale non-linear dynamic systems.
\end{IEEEbiography}

\end{document}